\documentclass[a4paper,11pt]{article}

\usepackage[T1]{fontenc}
\usepackage[utf8]{inputenc}
\usepackage{fullpage}
\usepackage{times}
\usepackage{microtype}
\usepackage{amsmath,amssymb,amsthm}
\usepackage{booktabs}
\usepackage{float}
\usepackage{algorithm}
\usepackage{algpseudocode}
\usepackage{enumerate}
\usepackage[small]{caption}
\usepackage{authblk}
\usepackage{natbib}
\usepackage[hidelinks]{hyperref}
\usepackage{cleveref}
\usepackage{xcolor}

\newtheorem{theorem}{Theorem}[section]

\newtheorem{proposition}[theorem]{Proposition}
\newtheorem{lemma}[theorem]{Lemma}

\theoremstyle{definition}
\newtheorem{definition}[theorem]{Definition}
\newtheorem{example}[theorem]{Example}

\allowdisplaybreaks

\title{\bf Strengthening Proportionality in Participatory Budgeting with Additive Utilities}

\author[1]{Tzeh Yuan Neoh}
\author[2]{Nicholas Teh}
\affil[1]{Harvard University, USA}
\affil[2]{University of Oxford, UK}

\date{\vspace{-10mm}}

\begin{document}

\maketitle

\begin{abstract}
Proportional representation is a central goal in participatory budgeting, where voters select public projects subject to a shared budget. Full justified representation (FJR) accounts for groups whose members may value different projects, but computing an FJR outcome is strongly NP-hard for additive utilities. We show that FJR up to one project can be achieved in polynomial time for arbitrary nonnegative additive utilities and project costs. Our approach extends Residual-Budget Greedy and satisfies a stronger fractional axiom that can also be verified in polynomial time. For approval utilities, this axiom coincides with FJR+, providing exact FJR in that setting. We further show that every feasible completion of the algorithm’s selected set preserves the representation guarantee, providing flexibility in allocating the remaining budget. In particular, we give a completion that ensures priceability at the original budget. For unit-cost projects, we strengthen the representation guarantee to the Droop quota and combine it with priceability whenever there are enough positively valued projects to fill the budget.\end{abstract}

\section{Introduction}
\label{sec:intro}
Participatory budgeting (PB) is a process that allows residents to select public projects subject to a common budget. In New York City, for example, residents vote on improvements to schools, parks, libraries, and streets using PB. A central concern in this setting is whether sufficiently large groups are represented \emph{proportionally}.
Approval utilities, in which voters indicate only which projects they approve, capture many PB settings well and have supported a substantial literature on \emph{proportional representation}. In practice, however, projects can differ substantially in how much voters value them: a family may value a playground more than a library renovation, while another resident has the opposite preference. Additive utilities provide a richer model by allowing voters to express both which projects they value and how much they value them. They also make proportional representation more subtle. Consider a neighbourhood whose residents value different local improvements. Its proportional share of the budget might fund a set that gives every member substantial total utility, although no single project is valuable to everyone. A representation guarantee should account for such groups by considering each member's utility for the whole proposed set.

\emph{Full justified representation (FJR)}, addresses precisely this distinction. Suppose that a group can afford a set of projects with its proportional share of the budget, and that every member obtains utility at least $\beta$ from that set. FJR requires that at least one member obtain utility at least $\beta$ from the selected outcome. The requirement applies to every such group and set. In contrast, \emph{extended justified representation (EJR)} uses common lower bounds on the utility of each individual project. The difference is between the minimum of the members' total utilities and the sum of their project-by-project minimum utilities. The former can be substantially larger.

There is a computational distinction as well. The Method of Equal Shares (MES) is computable in polynomial time and its outcome satisfies EJR up to one project for additive utilities \citep{PetersPierczynskiSkowron2021}. 
FJR can be satisfied by the Greedy Cohesive Rule, but that rule searches over voter groups and proposed sets. More recently, for approval utilities and unit costs, \citet{FrankPeters2026} obtain FJR by modifying the Greedy Justified Candidate Rule, while \citet{Ai2026} give a descending-budget rule. Thus, in the approval setting, strong FJR guarantees can be achieved efficiently.

What remains unclear is whether comparable guarantees are possible for general additive utilities. Exact FJR is computationally difficult: \citet{BaharavEtAl2026} prove that finding an FJR outcome is strongly NP-hard, even for for unit costs and integer utilities bounded by a fixed constant. Under their definition, no sequential rule satisfies EJR even with unit costs and utilities in $\{0,1,2,3\}$ \citep[Theorem~6.1]{BaharavEtAl2026}; since FJR implies EJR, this also rules out FJR for that class of rules. Moreover, checking whether an arbitrary approval committee satisfies FJR is coNP-complete \citep{KalayciLiuKempe2025}.

We therefore ask whether the one-project relaxation can be extended from EJR to FJR for additive utilities while retaining polynomial-time computation and verification. We answer the computation question posed by \citet[Section~4]{FrankPeters2026} and also give polynomial-time verification. The qualification “up to one project” is the same strict one used for additive EJR: either a group member attains the required utility, or adding a project from the group's proposed set takes that member strictly above it. Thus the loss is less than the value of one of the group's own projects.

\subsection{Our Contributions}
We study PB with nonnegative additive utilities and arbitrary nonnegative project costs, including approval utilities and unit costs as special cases.
\begin{enumerate}
\item We prove that FJR up to one project can be computed in polynomial time for additive utilities and arbitrary project costs. More strongly, the outcome satisfies a fractional axiom, denoted by FJR1+, that can be checked using at most $nm$ linear programs, where $n$ and $m$ are the numbers of voters and projects. The axiom specializes exactly to FJR+ defined for the setting of PB with approval utilities \citep{Teh2026}.

\item We prove that every feasible completion of the set selected by our algorithm satisfies FJR1+. We give a completion that is priceable at the original budget: voters can fully pay for the selected projects from equal shares of the available budget, paying only for projects they value positively, and for each unselected project, the remaining budgets of voters who value it positively total at most its cost. Alternatively, adding projects until no further project fits within the budget gives an exhaustive outcome that still satisfies FJR1+. The algorithm extends the Residual-Budget Greedy (RBG) algorithm of \citet{Teh2026}.

\item For unit-cost PB with budget $k$, we strengthen the additive guarantee from the Hare quota $n/k$ to the strict Droop quota $n/(k+1)$. If at least $k$ projects are each positively valued by some voter, we compute an outcome of size $k$ satisfying both Droop-FJR1+ and priceability. We obtain priceability by completing the selected set using sequential Phragm\'en \citep{BrillEtAl2024Phragmen}, extending the Hare completion of \citet{Teh2026} to the Droop quota. On approval ballots, the outcome satisfies Droop-FJR+ and priceability.
\end{enumerate}

\paragraph{Difference in analysis for additive utilities vs. approval utilities in the PB setting.}
Extending RBG from approval utilities to additive utilities raises two challenges. 
A project can be worth more than a voter's remaining utility requirement, so its uncapped offer can exceed the voter's budget. 
Moreover, utility levels need not be small integers, and an algorithm that visits all integer levels can take exponential time. 
We address the first issue by matching the offer cap with the one-project qualification in the axiom. 
For the second, we show that to verify FJR1+ for an outcome $W$, it suffices to test for fractional violations at target utility values of the form $u_i(W)+u_i(c)$ with $c\notin W$ and $u_i(c)>0$. 
We also show that computation can skip intervals in which no project is affordable. 
A utility-weighted budget inequality then proves representation for all feasible completions.
Simply treating every positively valued project as approved is not sufficient either. 
Consider one voter, four unit-cost projects of utilities $1,1,3,3$, and budget two. Approval RBG selects the two projects of utility one if they come first in the fixed project order. The voter can instead afford the two projects of utility three, giving target utility six, whereas adding either to the selected outcome gives utility only five. Thus, this extension violates FJR1.

\subsection{Related Work}
Justified representation and EJR were introduced for approval-based committee elections by \citet{AzizEtAl2017}, and \citet{AzizLeeTalmon2018} introduced proportional representation axioms for PB with approval utilities. \citet{PetersPierczynskiSkowron2021} study the additive utilities PB model and introduce additive EJR, EJR up to one project, and FJR. We use their definition in which the additional project must belong to the group's proposed set and must take the voter strictly above the target. \citet{BrillEtAl2023PB} study proportionality for broad classes of approval-based satisfaction functions. \citet{MasarikPierczynskiSkowron2024} develop a general framework for proportionality under feasibility constraints, including a base FJR axiom for monotone utilities and extensions to weighted candidates. \citet{Skowron2026} studies additive utility proportionality under general feasibility constraints, including PB, and its PB EJR comparison uses project-by-project minimum utilities. Our FJR1 target instead uses each voter's additive utility for the whole proposed set and becomes exact FJR on approval utilities.

A recurring distinction is between finding and checking proportional outcomes. EJR committees can be found in polynomial time, although checking EJR or PJR is coNP-complete \citep{AzizEtAl2018Complexity}. \citet{BrillPeters2023} introduce EJR+, a stronger approval axiom that is polynomial-time verifiable, and \citet{KraiczyElkind2024} give a fast EJR check for outcomes returned by the Adaptive Method of Equal Shares (AMES). \citet{KalayciLiuKempe2025} prove coNP-completeness for checking FJR and introduce full proportional justified representation (FPJR). \citet{Teh2026} introduces approval FJR+ using fractional voter weights, project weights, and assignments, with a polynomial-time linear programming check. Our FJR1+ retains that approval specialization, replaces the approval counting equation by an additive utility equation, and limits the use of an unselected project when it already takes a voter strictly above the target. Our proof reduces verification to levels of the form $u_i(W)+u_i(c)$. Unlike the AMES check, our linear programs apply to an arbitrary proposed outcome and do not need the choices made by our algorithm.

The algorithm combines ideas from several price-based and greedy rules. The descending target levels come from the Greedy Justified Candidate Rule \citep{BrillPeters2023}; equal initial budgets come from MES \citep{PetersSkowron2020}, whose additive PB extension uses utility-proportional payments capped by remaining budgets \citep{PetersPierczynskiSkowron2021}; and division by a voter's remaining representation requirement is due to \citet{Ai2026}. \citet{Teh2026} develops the last idea in RBG and proves that every completion of the partial committee selected in the approval setting satisfies FJR+. We combine the capped utility-proportional payments of additive MES with RBG's division by the remaining representation requirement, and prove FJR1+ for additive utilities, arbitrary costs, and every feasible completion. \citet{FrankPeters2026} also show that an MES committee need not admit an exact-FJR completion, so the every-completion conclusion here is tied to the RBG construction and the one-project guarantee.

\citet{BrillPeters2024} study completions of committees selected by price-based rules and their welfare and representation guarantees. The sub-core was introduced for PB by \citet{MunagalaShenWang2022}; we use the cost-based approval form from \citet{Teh2026}. For indivisible public goods with additive utilities, \citet{FainMunagalaShah2018} obtain additive approximations to the core, where the additive term is measured by the largest utility of an agent for an item. Our one-project comparison uses a project in the particular set proposed by the group, concerns FJR rather than the core, and becomes exact on approval utilities. Finally, \citet{CaseyElkind2026} introduce strict Droop versions of justified representation. We extend the fractional Droop guarantee to additive utilities and combine it with priceability.

\section{Preliminaries} \label{sec:prelim}

A \emph{participatory budgeting} (PB) instance consists of a set of \emph{voters} $N=\{1,\dots,n\}$ and a set of \emph{projects} $C$, with $|C|=m$ and $n,m\ge1$.
The available \emph{budget} is $B>0$, and every project $c$ has \emph{cost} $\mathrm{cost}(c)\ge0$. For $T\subseteq C$, let $\mathrm{cost}(T)=\sum_{c\in T}\mathrm{cost}(c)$. An \emph{outcome} is a set $W\subseteq C$; it is \emph{feasible} if $\mathrm{cost}(W)\le B$. Voter $i$ has nonnegative utilities $u_i(c)$ and additive utility $u_i(T)=\sum_{c\in T}u_i(c)$. 
For the computational results, the budget, costs, and utilities are rational numbers given in binary.
Following convention in prior work, we express costs in units of a voter's equal budget share $B/n$: the normalized cost is $\rho_c= n\mathrm{cost}(c)/B$, and each voter has one unit of budget. Thus feasibility means $\sum_{c\in W}\rho_c\le n$, and $T$ is affordable with the proportional budget of $S\subseteq N$ when $\sum_{c\in T}\rho_c\le |S|$.

\emph{Approval utilities} are the special case in which $u_i(c)\in\{0,1\}$. Denoting $A_i=\{c\in C:u_i(c)=1\}$, we have $u_i(W)=|A_i\cap W|$. The \emph{unit-cost} special case has $\mathrm{cost}(c)=1$ for all $c$ and budget $B=k$, where $k\in\{1,\dots,m\}$; selecting $k$ projects is equivalent to electing a committee of size $k$. We use the terms projects and outcomes throughout, including in this special case. A \emph{completion} of $P$ is a feasible superset of $P$. It need not spend the entire budget. An outcome is \emph{exhaustive} if no further project can be added while preserving feasibility.

\begin{definition}[Extended justified representation]
\label{def:ejr}
    A feasible outcome $W$ satisfies \emph{extended justified representation} (EJR) if, for every nonempty $S\subseteq N$ and every $T\subseteq C$ with $\sum_{c\in T}\rho_c\le |S|$, there is a voter $i\in S$ such that $u_i(W)\ge \sum_{c\in T}\min_{j\in S}u_j(c)$.
\end{definition}
This is the additive EJR definition of \citet{PetersPierczynskiSkowron2021}, written using the largest common lower bound for each project's utility. In contrast, FJR uses the minimum of the members' total utilities for $T$.

\begin{definition}[Full justified representation]
\label{def:fjr}
A nonempty group $S\subseteq N$ is \emph{weakly $(\beta,T)$-cohesive} if $\beta>0$, $T\subseteq C$, and $|S|\ge\sum_{c\in T}\rho_c,
 \quad u_i(T)\ge\beta$ for all $i \in S$.
A feasible outcome $W$ satisfies \emph{full justified representation} (FJR) if each such group contains a voter $i$ with $u_i(W)\ge\beta$.
\end{definition}

Our relaxation follows the treatment of additive EJR in \citet{PetersPierczynskiSkowron2021}. There, the additional project must come from the group's proposed set and must take the voter strictly above the target. These requirements make the one-project definition coincide with exact EJR on approval utilities and rule out comparisons with an unrelated high-utility project. We retain them and replace the EJR target $\sum_{c\in T}\min_{i\in S}u_i(c)$ by the FJR target $\min_{i\in S}u_i(T)$.

\begin{definition}[FJR up to one project]
\label{def:fjr1}
    A feasible outcome $W$ satisfies \emph{FJR up to one project} (FJR1) if every weakly $(\beta,T)$-cohesive group $S$ contains a voter $i$ such that $u_i(W)\ge\beta$ or $u_i(W)+u_i(c)>\beta$ for some $c\in T\setminus W$.
\end{definition}

The extra project is used only to compare utilities; $W\cup\{c\}$ need not be feasible. Requiring $c\in T\setminus W$ is important: an unrelated project of arbitrarily large utility should not satisfy a group's representation requirement. The strict inequality is important as well. With approval utilities, FJR1 is equivalent to FJR. Indeed, if FJR is violated, take the integer level $\beta=\min_{i\in S}u_i(T)$. Every member then has integer utility at most $\beta-1$, and adding one approved project cannot give utility strictly greater than $\beta$. The reverse implication holds directly.

The definition also has a quantitative interpretation. Suppose that $S$ is weakly $(\beta,T)$-cohesive and that $u_i(c)\le\varepsilon\beta$ for all $i\in S$ and $c\in T$, where $0<\varepsilon<1$. Then FJR1 gives some $i\in S$ with $u_i(W)>(1-\varepsilon)\beta$.
Thus the approximation approaches exact FJR when the target utility comes from many individually small projects. This is a consequence of the one-project guarantee, not a separate assumption of our algorithm.

For comparison, EJR up to one project (EJR1) uses the same one-project condition at the smaller target $\sum_{c\in T}\min_{i\in S}u_i(c)$ \citep{PetersPierczynskiSkowron2021}. It is therefore implied by FJR1. The implication is strict already on approval ballots; \Cref{ex:ejr} gives a small example that is also priceable.

\paragraph{Residual-Budget Greedy for approval utilities.}
For positive costs, the approval PB version of Residual-Budget Greedy (RBG) \citep[Section~6.2]{Teh2026} starts with a selected set $P=\varnothing$ and remaining budgets $b_i=1$. At descending integer target levels $h$, a voter with $|A_i\cap P|<h$ offers $b_i/(h-|A_i\cap P|)$ to each unselected project she approves; all other offers are zero. A project is affordable when its offers total at least $\rho_c$. RBG repeatedly selects an affordable project and pays its cost within these offers, then decreases the level when none remains. It returns the selected set, which may then be completed. Zero-cost projects can be included initially; we present the additive extension in \Cref{sec:algorithm} and prove that this initialization preserves its guarantee.

\section{A Polynomial-Time Verifiable Strengthening}
\label{sec:verification}

We extend the fractional interpretation of FJR+ for approval utilities from \citet{Teh2026}. Start from the definition of weak cohesiveness (\Cref{def:fjr}). Think of replacing $S$ by a group containing weight $z_i$ of each voter $i$, and $T$ by a collection containing weight $y_c$ of each project $c$; zero weight means that the voter or project is omitted. Since each voter has one unit of normalized budget, the group's size and budget $|S|$ become $\sum_i z_i$, while the cost $\sum_{c\in T}\rho_c$ becomes $\sum_c\rho_c y_c$. Likewise, the utility $u_i(T)$ becomes the weighted sum $\sum_c u_i(c)y_c$, and the requirement $u_i(T)\ge\beta$ becomes $\sum_c u_i(c)y_c\ge\beta z_i$ for each voter with $z_i>0$. Thus both the voter's budget contribution and her utility requirement are proportional to her weight. Taking weights one on $S$ and $T$, and zero elsewhere, recovers ordinary cohesiveness. These weights are used only to compare outcomes; the voters' equal budget shares do not change.

To describe a violation, each participating voter must have $u_i(W)<\beta$. Two further restrictions distinguish projects already selected from those not selected. A selected project $c$ can contribute at most $u_i(c)\min\{z_i,y_c\}$, so utility already received is counted at the same voter weight as the target. An unselected project can contribute $u_i(c)y_c$ only if $u_i(W)+u_i(c)\le\beta$; otherwise adding it already takes the voter strictly above the target. This restriction is applied separately to each voter: a project excluded for one voter may still be counted by another.

\begin{definition}[Fractional FJR up to one project]
\label{def:plus}
For an outcome $W$ and a level $\beta>0$, a \emph{fractional violation} consists of nonnegative voter weights $z_i$ and project weights $y_c$ satisfying $\sum_{i\in N}z_i>0$, $\sum_{i\in N}z_i\ge\sum_{c\in C}\rho_c y_c$, and, for every voter $i$ with $z_i>0$, $u_i(W)<\beta$ and $\sum_{c\in W}u_i(c)\min\{z_i,y_c\}
 +\sum_{{c\notin W: u_i(W)+u_i(c)\le\beta}}u_i(c)y_c
 \ge\beta z_i$.
A feasible outcome satisfies \emph{fractional FJR up to one project} (FJR1+) if it has no fractional violation at any real level $\beta>0$.
\end{definition}

The selected projects contribute at most $u_i(W)z_i<\beta z_i$, so each participating voter must count positive utility from an unselected project. In contrast, the weight of an unselected project need not be at most $z_i$, as in approval FJR+; \Cref{ex:strictplus} illustrates this with voter weights $1/2$ and a project weight of one. The one-project restriction is also necessary: removing it can make the fractional axiom unsatisfiable, even with two voters and three unit-cost projects; see \Cref{ex:necessary}. At integer levels on approval utilities, this restriction is automatic.

\paragraph{Linear programming formulation.}
For fixed $W$ and $\beta$, introduce assignments $x_{ic}$ indicating how much of project $c$ is counted toward voter $i$'s utility. Minimize $\sum_{c\in C}\rho_c y_c$ over nonnegative $z_i,y_c,x_{ic}$, subject to
\begin{equation}
 \sum_{i\in N}z_i=1
 \label{eq:normalize}
\end{equation}
and
\begin{subequations}\label{eq:assignments}
\begin{alignat}{2}
 z_i&=0 &\quad&\text{if }u_i(W)\ge\beta, \label{eq:active}\\
 x_{ic}&=0 &\quad&\text{if }u_i(c)=0, \label{eq:zero}\\
 \sum_{c\in C}u_i(c)x_{ic}&=\beta z_i &\quad&\text{for every }i, \label{eq:utility}\\
 x_{ic}&\le y_c &\quad&\text{for every }i,c, \label{eq:project}\\
 x_{ic}&\le z_i &\quad&\text{if }c\in W, \label{eq:selected}\\
 x_{ic}&=0 &\quad&\text{if }c\notin W\text{ and }u_i(W)+u_i(c)>\beta. \label{eq:oneproject}
\end{alignat}
\end{subequations}
Since projects are public goods, \eqref{eq:project} bounds each voter's assignment separately, not the sum of the assignments. A project's cost is counted only once.

These constraints give an equivalent test for a fractional violation. In one direction, the assignment bounds and \eqref{eq:utility} imply the utility inequality in \Cref{def:plus}. Conversely, for each voter with $z_i>0$, set $x_{ic}=\min\{z_i,y_c\}$ on positively valued selected projects and $x_{ic}=y_c$ on positively valued unselected projects with $u_i(W)+u_i(c)\le\beta$, setting all other assignments to zero. By the utility inequality, the assigned utility is at least $\beta z_i>0$. Decrease that voter's assignments proportionally until it equals $\beta z_i$; all upper bounds remain valid. Set every assignment of a voter with $z_i=0$ to zero. This proves equivalence with the weight conditions in \Cref{def:plus} and \eqref{eq:assignments}.

A common positive rescaling of all weights and assignments preserves these conditions, so dividing them by $\sum_i z_i$ gives \eqref{eq:normalize}. Thus a violation at $\beta$ exists exactly when the linear program is feasible and its optimum is at most one. The same rescaling also shows that restricting all weights to $[0,1]$, without imposing \eqref{eq:normalize}, would not change the axiom. We use the equivalent assignment formulation in the proofs below.

\begin{lemma}
\label{lem:levels}
If nonnegative weights and assignments satisfy \eqref{eq:assignments} at a level $\beta>0$, with $\sum_i z_i>0$, then the assignments can be changed so that the same constraints hold at a level in
\begin{equation}
 \{u_i(W)+u_i(c):i\in N,\ c\notin W,\ u_i(c)>0\}.
 \label{eq:levels}
\end{equation}
The voter and project weights are unchanged.
\end{lemma}

\begin{proof}
Fix such weights and assignments at $\beta$. Every voter with $z_i>0$ has a positive assignment to some unselected project. Otherwise, \eqref{eq:selected} would give
\[
 \beta z_i=\sum_{c\in W}u_i(c)x_{ic}
 \le u_i(W)z_i<\beta z_i.
\]
Let $h=\max\{u_i(W)+u_i(c):i\in N,\ c\notin W,\ x_{ic}>0\}$.
The maximum is well defined; \eqref{eq:zero} implies that all projects appearing in it have positive utility for their assigned voter. By \eqref{eq:oneproject}, $h\le\beta$. For every voter with $z_i>0$, the existence of an unselected positive assignment gives $u_i(W)<h$.

Keep $z_i,y_c$ unchanged and multiply all assignments by $h/\beta$. The utility equations now hold at level $h$, all upper bounds remain valid, and every positive unselected assignment satisfies $u_i(W)+u_i(c)\le h$. Thus \eqref{eq:assignments} holds at level $h$, with the same voter and project weights.
\end{proof}

\begin{proposition}
\label{prop:relations}
For additive utilities, FJR1+ implies FJR1. On approval utilities, it coincides with approval FJR+.
\end{proposition}

\begin{proof}
\emph{The ordinary guarantee.}
Suppose that $W$ violates FJR1 for a weakly $(\beta,T)$-cohesive group $S$. Set $z_i=1$ for $i\in S$ and zero otherwise, and $y_c=1$ for $c\in T$ and zero otherwise. Each member has $u_i(W)<\beta$ and $u_i(W)+u_i(c)\le\beta$ for every $c\in T\setminus W$. Thus the left-hand side of the utility inequality in \Cref{def:plus} equals $u_i(T)\ge\beta$. Since $S$ is nonempty and can afford $T$, the weight conditions also hold, so these weights form a fractional violation, a contradiction.

\emph{Approval utilities.}
By \Cref{lem:levels}, every fractional violation has a level in \eqref{eq:levels}. All these levels are integers between $1$ and $m$, since an unselected approved project gives $u_i(W)+1\le |W|+1\le m$. At an integer level $h$, a participating approval voter has $u_i(W)\le h-1$, so \eqref{eq:oneproject} imposes no additional restriction. The remaining constraints are exactly those of FJR+ for PB with approval utilities \citep[Definition~6.1]{Teh2026}, using the same inequalities also for zero-cost projects. Conversely, every approval FJR+ violation at an integer level satisfies \eqref{eq:oneproject}.
\end{proof}

Then, our main result is as follows.

\begin{theorem}
\label{thm:verify}
FJR1+ can be verified in polynomial time using at most $n(m-|W|)$ linear programs, each with $O(nm)$ variables and constraints.
\end{theorem}

\begin{proof}
First check feasibility of $W$. For each level $h$ in \eqref{eq:levels}, solve the linear program above with $\beta=h$. By the assignment formulation and \Cref{lem:levels}, a violation exists exactly when at least one program is feasible and has optimum at most one. The objective is nonnegative, so every feasible program has a finite attained minimum, even if some project costs are zero. There are at most $n(m-|W|)$ levels, and each program has $O(nm)$ variables and constraints. Hence verification takes polynomial time. Normalizing voter weight rather than project cost also covers violations involving only zero-cost projects.
\end{proof}

The verification theorem concerns the fractional axiom, not an exact test for ordinary FJR1. On approval utilities, FJR1+ implies FJR, but an outcome satisfying FJR need not satisfy FJR1+. This is the strict strengthening already established for approval FJR+ by \citet{Teh2026}; \Cref{ex:strictplus} illustrates that existing distinction. The verification procedure applies to an arbitrary proposed outcome, without needing the sequence of choices made by the algorithm.

\section{Efficient Computation of the Stronger Guarantee}
\label{sec:algorithm}

We extend the approval RBG rule of \cite{Teh2026} to the setting with additive utilities. As in additive MES, the offer is proportional to the voter's utility for the project until it reaches her remaining budget \citep{PetersPierczynskiSkowron2021}; here the proportionality factor also depends on her remaining utility requirement at the current target level.

By rescaling utilities and target levels by a common positive factor, we may assume that utilities are integers; this preserves FJR1 and FJR1+. 
Let $U=\max_i u_i(C)$. Initialize $P=\{c\in C:\rho_c=0\}$ and budgets $b_i=1$, assigning zero payments to these projects. If $U=0$, return this set. Zero-cost projects therefore contribute to $u_i(P)$ from the start, without reducing any budget.

The algorithm processes utility levels $h=U,U-1,\dots,1$. At level $h$, a voter below the level offers
\begin{equation}
 o_{ic}(h)=
 \begin{cases}
 \displaystyle\min\left\{b_i,\frac{b_i u_i(c)}{h-u_i(P)}\right\},&u_i(P)<h,\\[1ex]
 0,&u_i(P)\ge h
 \end{cases}
 \quad \text{for } c\notin P.
 \label{eq:offers}
\end{equation}
The dependence on the current $P$ and $b$ is implicit. The ratio $u_i(c)/(h-u_i(P))$ is the fraction of the voter's remaining utility requirement provided by $c$. The voter offers the same fraction of her remaining budget, but never more than that budget. Voters already at the level do not contribute.

A project is affordable at level $h$ when $\sum_i o_{ic}(h)\ge\rho_c$. The algorithm repeatedly buys an affordable project, charging payments within the offers, until no project is affordable at that level. Fix project and voter orders to make all choices deterministic.

\begin{algorithm}[ht]
\caption{Residual-Budget Greedy for additive utilities}
\label{alg:rbg}
\begin{algorithmic}[1]
\Require Integer utilities, $U=\max_i u_i(C)$, normalized costs $\rho_c\ge0$, fixed project and voter orders
\State $P\gets\{c\in C:\rho_c=0\}$; $b_i\gets1$ for each $i\in N$; assign zero payments to $P$
\For{$h=U,U-1,\dots,1$}
 \While{some $c\notin P$ satisfies $\sum_i o_{ic}(h)\ge\rho_c$}
  \State Choose the first such project $c$; set $r\gets\rho_c$
  \For{each voter $i$ in the fixed order}
   \State $p_i\gets\min\{o_{ic}(h),r\}$; $r\gets r-p_i$
  \EndFor
  \State $b_i\gets b_i-p_i$ for each $i$; $P\gets P\cup\{c\}$
 \EndWhile
\EndFor
\State \Return $P$ and the payments
\end{algorithmic}
\end{algorithm}

\begin{theorem}
\label{thm:compute}
For additive utilities and arbitrary project costs, \Cref{alg:rbg} has a polynomial-time implementation. It returns a feasible set $P$ such that every feasible completion of $P$ satisfies FJR1+.
\end{theorem}

The proof adapts the remaining-budget comparison in \citet[Theorem~6.4]{Teh2026}. Additive utilities require the remaining assigned utility $t_i$ in addition to the remaining utility requirement $a_i$. The one-project condition and the budget cap make this comparison valid when one project can exceed the remaining utility requirement. For a supposed fractional violation, the ratio $t_i/a_i$ is at least the voter's weight. Multiplying this ratio by the voter's remaining budget gives an amount that can be summed across the group. A purchase made at or above the target level cannot reduce this sum by more than the corresponding fractional project cost. At the target level, a violation would therefore require some project to remain affordable.

For the next two lemmas, fix an outcome $W$, a level $\beta>0$, and nonnegative weights and assignments satisfying \eqref{eq:assignments}, with $\sum_i z_i>0$. We consider a current selected set $P\subseteq W$ and nonnegative remaining budgets $b_i$.

For every voter with $z_i>0$, define her remaining utility requirement and remaining assigned utility by $a_i=\beta-u_i(P)$ and $t_i=\sum_{c\notin P}u_i(c)x_{ic}$.
These quantities change with $P$. Since $u_i(W)<\beta$ and $P\subseteq W$, we have $a_i>0$. Moreover,
\begin{equation}
 t_i=\beta z_i-\sum_{c\in P}u_i(c)x_{ic}
 \ge (\beta-u_i(P))z_i=a_i z_i,
 \label{eq:remaining}
\end{equation}
where the equality follows from \eqref{eq:utility} and the inequality uses \eqref{eq:selected}. Consider
\begin{equation}
 \Phi(P,b)=\sum_{i:z_i>0}\frac{b_i t_i}{a_i}
          -\sum_{c\notin P}\rho_c y_c.
 \label{eq:potential}
\end{equation}
\begin{lemma}[Remaining-budget comparison]
\label{lem:budget}
Initially, $\Phi(P,\mathbf{1})\ge\sum_i z_i-\sum_{c\in C}\rho_c y_c$. Every purchase at a level $h\ge\beta$ of a project in $W\setminus P$ leaves $\Phi$ unchanged or increases it.
\end{lemma}

\begin{proof}
Initially, $b_i=1$ and all projects in $P$ have zero cost. Thus \eqref{eq:remaining} gives
\[
 \Phi(P,\mathbf{1})\ge\sum_i z_i-\sum_{c\in C}\rho_c y_c.
\]
When all costs are positive, $P=\varnothing$, so \eqref{eq:utility} gives $t_i=\sum_{c\in C}u_i(c)x_{ic}=\beta z_i$ and $a_i=\beta$.

Suppose project $d\in W\setminus P$ is purchased at level $h\ge\beta$ with payments $p_i$. For each voter with $z_i>0$, the fact that $d\in W\setminus P$ gives us
\[
 a_i>u_i(d),\quad x_{id}\le z_i,\quad
 p_i\le\frac{b_i u_i(d)}{h-u_i(P)}
      \le\frac{b_i u_i(d)}{a_i}.
\]
After the purchase, $a_i,t_i,b_i$ become $a_i-u_i(d)$, $t_i-u_i(d)x_{id}$, and $b_i-p_i$, respectively. Direct expansion gives
\begin{equation}
\begin{aligned}
 \frac{(b_i-p_i)(t_i-u_i(d)x_{id})}{a_i-u_i(d)}
       -\frac{b_i t_i}{a_i}+p_i x_{id} =
 \frac{(b_i u_i(d)-a_i p_i)(t_i-a_i x_{id})}
      {a_i(a_i-u_i(d))}\ge0.
\end{aligned}
\label{eq:identity}
\end{equation}
The denominator is positive. The first factor in the numerator is nonnegative by the payment bound, and the second by \eqref{eq:remaining} and $x_{id}\le z_i$.

Thus the first sum in \eqref{eq:potential} decreases by at most $\sum_i p_i x_{id}$. Assignments of voters with $z_i=0$ are zero by \eqref{eq:zero} and \eqref{eq:utility}, so including these voters in the sum changes nothing. Since $x_{id}\le y_d$ and $\sum_{i\in N}p_i=\rho_d$, $\sum_i p_i x_{id}\le\rho_d y_d$.
The second sum in \eqref{eq:potential} decreases by exactly $\rho_d y_d$. Hence $\Phi$ does not decrease during any purchase at a level at least $\beta$.
\end{proof}

\begin{lemma}[Affordability at the target level]
\label{lem:affordability}
If $P$ contains all zero-cost projects and no project outside $P$ is affordable at level $\beta$, then $\Phi(P,b)<0$.
\end{lemma}

\begin{proof}
Every $c\notin P$ has $\rho_c>0$, since all zero-cost projects were selected initially. If $\sum_{c\notin P}\rho_c y_c=0$, nonnegativity would give $y_c=0$ for every remaining project. Then \eqref{eq:project} would give $x_{ic}=0$ on all of them, and hence $t_i=0$, contradicting $t_i\ge a_i z_i>0$ for any voter with $z_i>0$. Therefore $\sum_{c\notin P}\rho_c y_c>0$.

For every positive assignment $x_{ic}$ with $c\notin P$, the budget cap in \eqref{eq:offers} does not reduce the uncapped offer at level $\beta$. Indeed, if $c\notin W$, then \eqref{eq:oneproject} gives $u_i(c)\le\beta-u_i(W)\le a_i$. Equality is allowed: if $u_i(c)=a_i$, the uncapped offer is exactly $b_i$. If $c\in W\setminus P$, additivity gives $u_i(c)\le u_i(W)-u_i(P)<a_i$. Therefore, on each such assignment, $o_{ic}(\beta)=\frac{b_i u_i(c)}{a_i}$.
Using $x_{ic}\le y_c$, we obtain
\begin{equation}
 \sum_{i:z_i>0}\frac{b_i t_i}{a_i}
 =\sum_{\substack{c\notin P\\y_c>0}}y_c
      \sum_{i:z_i>0}\frac{x_{ic}}{y_c}\frac{b_i u_i(c)}{a_i} \le\sum_{c\notin P}y_c\sum_i o_{ic}(\beta)
 <\sum_{c\notin P}\rho_c y_c.
\label{eq:contradiction}
\end{equation}
The final inequality is strict because no unselected project is affordable by assumption, and the sum on the right is positive. Hence $\Phi(P,b)<0$.
\end{proof}

\begin{proof}[Proof of \Cref{thm:compute}]
Every selected project is fully paid, and $0\le p_i\le b_i$. Consequently, the total normalized cost of $P$ is at most $n$; hence $P$ is feasible.

Fix a feasible completion $W$ of the final selected set, and suppose that $W$ has a fractional violation. By \Cref{lem:levels}, choose its level $\beta$ from \eqref{eq:levels}, with weights and assignments satisfying the weight conditions in \Cref{def:plus} and \eqref{eq:assignments}. This is an integer between $1$ and $U$. Throughout the execution, the current selected set satisfies $P\subseteq W$, so the two lemmas apply to this fixed violation. By the affordability condition and \Cref{lem:budget}, $\Phi(P,b)\ge0$ initially and throughout all purchases at levels at least $\beta$. At the end of level $\beta$, all zero-cost projects are in $P$ and no unselected project is affordable. Thus, \Cref{lem:affordability} gives $\Phi(P,b)<0$, a contradiction.

It remains to implement the descending levels efficiently. At a fixed $P,b$, divide the remaining integer levels into intervals on which the set $\{i:u_i(P)<h\}$ is constant. There are at most $n+1$ intervals. Within each interval, every offer is nonincreasing in $h$, so affordability can be located by binary search. Intervals with no affordable project are skipped. Each purchase changes the current state, after which the search is repeated, with the current level as an upper bound. There are at most $m$ purchases. Appendix~\ref{app:runtime} gives the complete implementation. This proves the claimed polynomial running time.
\end{proof}

The guarantee is independent of the project order and of the choice of payments within the offers. The common order is only a way to specify a deterministic implementation. Our proof also does not ask a completion to follow the algorithm's offers: any additional feasible projects are allowed.

\section{Priceability at the Original Budget}
\label{sec:priceability}

The representation guarantee alone does not determine how project costs should be shared. We consider another property, \emph{priceability}, which adds an equal-budget interpretation of the outcome and guides our choice of additional projects.
Intuitively, priceability requires that selected projects be paid for from equal voter budgets, with payments restricted to projects that the voter values positively, and that no unselected project can collect more than its cost from its interested voters' remaining budgets. We use the standard PB definition \citep{PetersSkowron2020,PetersPierczynskiSkowron2021}.

\begin{definition}[Priceability]
\label{def:priceability}
A feasible outcome $W$ is \emph{priceable} if there exist $B'\ge B$ and payments $p_i(c)\ge0$ such that
\begin{enumerate}[(i)]
    \item $p_i(c)=0$ if $u_i(c)=0$ or $c\notin W$,
    \item $\sum_{c\in W}p_i(c)\le B'/n$ for every $i$,
    \item $\sum_i p_i(c)=\mathrm{cost}(c)$ for every $c\in W$,
    \item $\sum_{i:u_i(c)>0}(\frac{B'}n-\sum_{w\in W}p_i(w)) \le\mathrm{cost}(c)$ for every $c\notin W$.
\end{enumerate}
If these conditions hold with $B'=B$, we say that $W$ is priceable at the original budget.
\end{definition}

For approval utilities, we use the following cost-based PB form of the sub-core from \citet{Teh2026}, based on the sub-core introduced by \citet{MunagalaShenWang2022}. It excludes nonempty $S\subseteq N$ and $T\subseteq C$ such that
\begin{equation}
 \sum_{c\in T}\rho_c\le |S|,
 \quad A_i\cap W\subsetneq A_i\cap T\quad \text{for all }i\in S.
 \label{eq:subcore}
\end{equation}
Thus a group cannot use its proportional budget to retain all its members' approved selected projects and add an approved project for every member.

\begin{theorem}
\label{thm:priceability}
For additive utilities and arbitrary project costs, a feasible outcome satisfying FJR1+ and priceability at the original budget can be computed in polynomial time. For approval utilities, the same outcome satisfies the cost-based sub-core.
\end{theorem}

\begin{proof}
The completion step below is the approval-PB continuation used by \citet{Teh2026}. It depends only on which projects each voter values positively, so it applies to the payments produced by additive RBG. The FJR1+ conclusion here follows from \Cref{thm:compute}, which applies to every feasible completion.

Run \Cref{alg:rbg}. Continue selecting any unselected project $c$ for which $\sum_{i:u_i(c)>0}b_i\ge\rho_c$, paying its normalized cost from those voters' remaining budgets. At most $m$ projects are selected, no budget becomes negative, and all payments are to positively valued projects. Let $W$ be the final set. It is a feasible completion and hence satisfies FJR1+ by \Cref{thm:compute}. At termination, every $c\notin W$ satisfies
\begin{equation}
 \sum_{i:u_i(c)>0}b_i<\rho_c.
 \label{eq:unspent}
\end{equation}
Multiplying all payments and remaining budgets by $B/n$ gives us \Cref{def:priceability} with $B'=B$. Fixed voter and project orders give a polynomial-time implementation.

Suppose, in the approval case, that $S,T$ satisfy \eqref{eq:subcore}. The selected projects paid for by members of $S$ all belong to $T$, so their total expenditure is at most $\sum_{c\in T\cap W}\rho_c$. Every member approves a project in $T\setminus W$. Summing \eqref{eq:unspent} over these unselected projects bounds the members' total remaining budget strictly below $\sum_{c\in T\setminus W}\rho_c$. Adding expenditure and remaining budget gives $|S|<\sum_{c\in T}\rho_c$, a contradiction.
\end{proof}

The original budget in this theorem is not increased for the representation guarantee or for the payments. The resulting outcome need not be exhaustive. Alternatively, one may complete the RBG outcome by repeatedly adding any budget-feasible project; \Cref{thm:compute} preserves FJR1+. These are different choices: exhaustiveness and priceability cannot in general be satisfied simultaneously with arbitrary costs \citep{PetersPierczynskiSkowron2021}. 

\section{Droop Representation and Priceability for Unit Costs}
\label{sec:droop}

We next ask whether the stronger representation requirement based on the strict Droop quota can also be combined with priceability. We consider unit-cost PB with $B=k$ and arbitrary nonnegative additive utilities.

The standard Droop quota refers to a fixed committee size $k$ \citep{CaseyElkind2026}; in the unit-cost model, the $+1$ in $k+1$ corresponds to one additional project. For arbitrary costs, the budget does not specify such a committee size. Simply replacing $|T|/(k+1)$ by $\mathrm{cost}(T)/(B+1)$ would make the requirement depend on the cost units: rescaling all costs and $B$ leaves $\rho_c$ and feasibility unchanged but need not preserve this ratio. Thus, the standard Droop condition is not directly defined for arbitrary-cost PB; an analogue would require a separate choice of definition.

Since $\rho_c=n/k$, the affordability condition for weak cohesiveness is $|S|\ge(n/k)|T|$. Following the strict Droop convention \citep{CaseyElkind2026}, a feasible outcome $W$ satisfies \emph{Droop-FJR1} if, for every nonempty group $S\subseteq N$, every $T\subseteq C$, and every $\beta>0$ satisfying $|S|>\frac{n}{k+1}|T|$, $u_i(T)\ge\beta$ for every $i \in S$, there is a voter $i\in S$ such that $u_i(W)\ge\beta$ or $u_i(W)+u_i(c)>\beta$ for some $c\in T\setminus W$. Only the affordability condition in \Cref{def:fjr1} changes. 

On approval utilities, this is equivalent to requiring some $i\in S$ with $|A_i\cap W|\ge\beta$ whenever the displayed conditions hold with $|A_i\cap T|\ge\beta$ and integer $\beta>0$, by the same integrality argument as for FJR1.

For the fractional version, keep the utility requirements of FJR1+ (\Cref{def:plus}) and replace its weight conditions by
\begin{equation}
 \sum_i z_i>0,
 \quad \sum_i z_i>\frac{n}{k+1}\sum_c y_c.
 \label{eq:droopweight}
\end{equation}
A feasible outcome with no such violation satisfies \emph{Droop-FJR1+}. The same choice of voter and project weights as in \Cref{prop:relations} shows that it implies Droop-FJR1. It also implies the Hare version: the positive project weight in any Hare violation makes the inequality in \eqref{eq:droopweight} strict. On approval utilities, it is exactly Droop-FJR+.

The project price $n/(k+1)$ comes from the descending-budget Droop rule of \citet{Ai2026}. \citet{Teh2026} proves Droop-FJR+ for approval RBG and combines its Hare version with sequential Phragm\'en to obtain FJR+ and priceability. We extend the representation argument to additive utilities through the utility-weighted offers and condition~\eqref{eq:oneproject}.

\begin{theorem}
\label{thm:droop}
For PB with additive utilities, unit costs, and budget $B=k$, an outcome of size $k$ satisfying Droop-FJR1+ can be computed in polynomial time. This property can also be verified in polynomial time. If at least $k$ projects are positively valued by at least one voter, the outcome can simultaneously be chosen to be priceable. In particular, for approval utilities there is a polynomial-time rule satisfying Droop-FJR+, priceability, and the sub-core under this condition.
\end{theorem}

We give an intuition for the construction and the reason for the stronger quota. Run \Cref{alg:rbg} with price $n/(k+1)$ for every project, stopping purchases when $k$ projects have been selected. The invariant from \eqref{eq:potential} now starts strictly positive for a putative Droop violation. If the relevant level ends before $k$ projects have been selected, the previous affordability argument applies. If the algorithm reaches $k$ projects first, the total remaining voter budget is exactly $n/(k+1)$. The one-project condition then bounds the first sum of \eqref{eq:potential} by the second, again giving a contradiction. Thus every size-$k$ completion retains Droop-FJR1+.

For priceability, divide the payments by the project price and interpret them as unit project loads. Each voter's initial load is at most $(k+1)/n$. Complete by sequential Phragm\'en \citep{BrillEtAl2024Phragmen}, starting from the RBG loads as in \citet{Teh2026}. At the end, either all loads remain at most $(k+1)/n$, in which case the entire electorate has exactly one unit of unused capacity at that level, or the last Phragm\'en choice determines the largest load. In the first case no unselected project's voters have more than one unit of unused capacity; in the second, the same bound follows from Phragm\'en's choice. These loads give the required payments. Appendix~\ref{app:droop} supplies the full proof of our result

Unlike \Cref{thm:priceability}, the unit-cost completion uses the usual definition of priceability, allowing $B'>k$. The requirement that at least $k$ projects receive positive utility is necessary: an outcome of size $k$ otherwise contains a project for which every payment must be zero. Finally, this is a strict improvement over combining Hare-FJR+ with priceability, not an implication of that combination; \Cref{ex:droop} exhibits a priceable Hare-FJR+ outcome that violates Droop-FJR+.

\section{Discussion}
\label{sec:discussion}

The main conclusion from our work is that weakly cohesive groups can be represented efficiently with additive utilities while allowing a loss of less than one of their own projects. 
The same condition that permits this relaxation gives rise to a polynomial-time verifiable fractional axiom. The guarantee is exact on approval ballots and approaches exact FJR when individual projects are small relative to a group's utility requirement. The algorithm separates representation from the choice of additional projects, permitting either arbitrary feasible completions or priceable completions.

A natural direction for future work is to study analogous guarantees in the temporal PB setting, building on work on long-term PB \citep{LacknerMalyRey2021}. In fact, weaker versions of existing open questions in temporal PB are still open in temporal voting and sequential decision making \citep{bulteau2021jrperpetual,elkind2024temporalelections,ElkindObraztsovaTeh2024TemporalFairness,elkind2025temporalchores,chandak2026proportional,teh2026price}. For temporal voting with approval ballots and preferences known in advance, FJR outcomes always exist, but finding them in polynomial time remains open \citep{PhillipsEtAl2026Temporal}. Extending our fractional approach to this setting, as suggested for FJR+ by \citet{Teh2026}, would require accounting for the restriction to one alternative per round.

\subsection*{Declaration of generative AI use}

The authors used GPT-6 as a research assistant when exploring proof ideas and revising the
exposition. They checked the mathematical arguments and take full responsibility for the paper.

\bibliographystyle{plainnat}
\bibliography{bib}

\appendix

\section{Polynomial-Time Implementation}
\label{app:runtime}

If $U=0$, no fractional violation is possible by \eqref{eq:utility}, and the algorithm returns its initial zero-cost projects.

\paragraph{Skipping empty levels.}
Affordability need not be monotone over all levels: when $h$ decreases to $u_i(P)$, voter $i$ stops contributing. For example, with one voter, $\rho_c=b_1=1/2$, and $u_1(P)=u_1(c)=1$, project $c$ is affordable at $h=2$ but not at $h=1$.

At fixed $P,b$, let $H$ be the largest level still to be processed, initially $U$. Partition $[1,H]$ at the integers $u_i(P)+1$ that lie in it. On each resulting interval $[L,R]$, the set $\{i:u_i(P)<h\}$ is fixed, and every offer in \eqref{eq:offers} is nonincreasing in $h$. Thus, within that interval, affordability at $h$ implies affordability at every smaller level.

Inspect intervals in decreasing order. If no project is affordable at $L$, skip the interval. Otherwise binary-search for its largest affordable integer level $h$, select the first affordable project, and charge payments as in \Cref{alg:rbg}. Repeat at the updated $P,b$, setting $H=h$ so further purchases at that level are considered but larger levels are not revisited. If no interval contains an affordable project, terminate. This makes the same choices as the full descending execution.

There are at most $m$ purchases and one final search. Each search considers at most $n+1$ intervals and uses $O(n+\log(U+1))$ affordability tests: binary search is needed only in the first interval containing an affordable project. Each test evaluates at most $nm$ offers; sorting the endpoints takes $O(n\log(n+1))$ comparisons. This gives a polynomial-time implementation.

All arithmetic uses numbers of polynomial bit length. Multiplying utilities by the product of their input denominators makes them integers of polynomial bit length, so $\log(U+1)$ is polynomial in the input length. Initially, the normalized costs and remaining budgets have a common denominator of polynomial bit length. At each purchase, this denominator need only be multiplied by the positive integers $h-u_i(P)$ for the contributing voters. Each is at most $U$, so the at most $m$ purchases add $O(mn\log(U+1))$ denominator bits. The budgets lie in $[0,1]$, so their numerators also have polynomial bit length.

\section{Proof of the Theorem \ref{thm:droop}}
\label{app:droop}

Write $q=n/(k+1)$ and run \Cref{alg:rbg} with every project price equal to $q$, stopping purchases once $k$ projects have been selected. The utility scaling in \Cref{sec:algorithm} and the implementation in Appendix~\ref{app:runtime} apply without change. Let $W$ be any size-$k$ completion of the selected set.

\paragraph{Representation.}
Suppose that $W$ has a fractional Droop violation. \Cref{lem:levels} leaves the voter and project weights unchanged, so take its level $\beta$ to be an integer between $1$ and $U$. Use $a_i,t_i$ as defined in \Cref{sec:algorithm}, and consider
\[
 \Phi(P,b)=\sum_{i:z_i>0}\frac{b_i t_i}{a_i}
                 -q\sum_{c\notin P}y_c.
\]
By \eqref{eq:droopweight}, $\Phi(\varnothing,\mathbf{1})>0$. Identity~\eqref{eq:identity} proves that it cannot decrease during purchases at levels $h\ge\beta$.

If level $\beta$ ends before $k$ projects have been purchased, \Cref{lem:affordability}, with every project price replaced by $q$, contradicts $\Phi>0$. Otherwise the algorithm reaches $k$ projects before that level has ended. At that moment $P=W$ and $\sum_i b_i=n-kq=q$.
For each voter with $z_i>0$ and each positive remaining assignment, \eqref{eq:oneproject} gives $u_i(c)\le\beta-u_i(W)=a_i$. Therefore
\[
 t_i=\sum_{c\notin W}u_i(c)x_{ic}
 \le a_i\sum_{c\notin W}y_c.
\]
It follows that
\[
 \sum_{i:z_i>0}\frac{b_i t_i}{a_i}
 \le\left(\sum_{c\notin W}y_c\right)\sum_i b_i
 =q\sum_{c\notin W}y_c,
\]
again contradicting $\Phi>0$. Thus every size-$k$ completion satisfies Droop-FJR1+.

\paragraph{Verification and approval utilities.}
For each level $h$ in \eqref{eq:levels}, minimize $q\sum_c y_c$ over nonnegative variables subject to \eqref{eq:assignments} at $\beta=h$ and \eqref{eq:normalize}. A fractional Droop violation exists exactly when one of these programs is feasible and has optimum strictly less than one. As in \Cref{thm:verify}, this can be determined in polynomial time by linear programming.

For approval utilities, \eqref{eq:oneproject} is automatic at integer levels. Moreover, $h\sum_i z_i\le n\sum_c y_c$, together with \eqref{eq:droopweight}, implies $h<k+1$. Thus the relevant integer levels are $1,\dots,k$, and the axiom is precisely Droop-FJR+.

\paragraph{Completion by sequential Phragm\'en.}
We use the standard sequential-Phragm\'en load update \citep{BrillEtAl2024Phragmen}, starting from the RBG payments as in \citet{Teh2026}. Assume at least $k$ projects are positively valued. Divide each payment made by the algorithm by $q$. Denote the resulting loads by $p_i(c)$ and the total load of voter $i$ by $L_i=\sum_{c\in P}p_i(c)$. Each selected project has total load one, positively valued projects are the only ones with positive payments, and initially $0\le L_i\le\frac1q=\frac{k+1}{n}$.
The same $p_i(c)$ will be used as payments for unit-cost projects.

Until $k$ projects have been selected, for each unselected project $c$ with at least one voter assigning positive utility, let $\tau_c$ be the unique solution of
\begin{equation}
 \sum_{i:u_i(c)>0}\max\{0,\tau_c-L_i\}=1.
 \label{eq:phragmen}
\end{equation}
Select a project minimizing $\tau_c$, breaking ties by a fixed order. Give voter $i$ load $p_i(c)=\max\{0,\tau_c-L_i\}$ if $u_i(c)>0$, and zero otherwise, and increase $L_i$ accordingly. There is always an eligible project while fewer than $k$ have been selected. Each step assigns one unit of total load, and all loads are nondecreasing.

Let $W$ be the resulting outcome of size $k$, and set $M=\max \{\frac{k+1}{n},\ \max_i L_i \}$.
We prove that every unselected project satisfies
\begin{equation}
 \sum_{i:u_i(c)>0}(M-L_i)\le1.
 \label{eq:loadprice}
\end{equation}
If $M=(k+1)/n$, every summand $M-L_i$ is nonnegative and $\sum_i(M-L_i)=nM-k=1$.
The sum over any subset of voters is therefore at most one.

Otherwise $M> (k+1)/n$. This largest load was not present initially; it was created by the completion. Project thresholds in \eqref{eq:phragmen} cannot decrease when voter loads increase. Consequently, the successive selected thresholds are nondecreasing, and the last selected threshold is exactly $M$. Immediately before the last selection, every project that will remain unselected has threshold at least $M$. Equation~\eqref{eq:phragmen} then bounds its voters' total load increase at level $M$ by one. Their loads only increase after that selection, so their final sum in \eqref{eq:loadprice} is also at most one. A project with no positive utility has an empty sum and satisfies the inequality as well.

Now choose $B'=nM$. Each voter has budget $M$, each selected project receives payment one, and the payment restrictions hold. Inequality~\eqref{eq:loadprice} gives the condition for every unselected project. Since $B'\ge k+1>k$, these payments satisfy \Cref{def:priceability}. The completion preserves the selected set, so its representation guarantee follows from the first part of the proof.

\paragraph{The approval sub-core.}
Suppose a nonempty approval group $S$ and set $T$ satisfy $|T|\le k|S|/n$ and $A_i\cap W\subsetneq A_i\cap T$ for every member. These members spend only on projects in $T\cap W$. Also, each member approves a project in $T\setminus W$. Summing their expenditure and the remaining-budget inequalities \eqref{eq:loadprice} therefore gives $|S|M\le |T|$.
But $M\ge(k+1)/n$ gives $|S|M>k|S|/n\ge |T|$, a contradiction. In fact, the same argument excludes such strict-inclusion deviations whenever $|S|>n|T|/(k+1)$.

\paragraph{Running time of the completion.}
To compute $\tau_c$, sort the loads of voters with $u_i(c)>0$. On the interval where exactly the first $r$ sorted loads are below the threshold, \eqref{eq:phragmen} gives us $\tau_c=\frac{1+\sum_{j=1}^{r}L_{i_j}}{r}$.
Scanning the sorted loads finds the correct interval, including ties. This takes polynomial time for each project. There are at most $k$ completion steps, so the completion takes polynomial time.
Dividing the RBG payments by $q$ preserves polynomial bit length, and each completion step multiplies a common denominator of the loads by an integer at most $n$, so all arithmetic remains polynomial in the input length.

\section{Examples Clarifying the Guarantees}
\label{app:examples}

\begin{example}[EJR and priceability do not imply FJR1]
\label{ex:ejr}
Let $n=k=3$, with unit costs, and approval ballots
\[
 A_1=\{a,b,d_1\},\quad
 A_2=\{a,c,d_2\},\quad
 A_3=\{b,c,d_3\}.
\]
Consider $W=\{d_1,d_2,d_3\}$. Every voter has utility one. A singleton can afford at most one project; a pair has only one commonly approved project; and all three voters have no commonly approved project. Thus $W$ satisfies EJR, and hence its approval-equivalent one-project version. It is also priceable at budget three: each voter pays one for her selected project and has no money left.

However, $S=N$ can afford $T=\{a,b,c\}$, and every member has utility two from $T$. Adding any one project from $T$ to $W$ gives a member utility at most two, not strictly more. Therefore $W$ violates FJR1. This illustrates the importance of evaluating each member's total utility for the proposed set, rather than only projects that the group values in common.
\end{example}

\begin{example}[Why condition~\eqref{eq:oneproject} is needed]
\label{ex:necessary}
Consider the unit-cost instance from \citet[Footnote~8]{PetersPierczynskiSkowron2021}, with $n=2$, $B=2$, and projects $a,b,c$ of utilities
\[
\begin{array}{c|ccc}
 &a&b&c\\ \hline
 u_1&2&3&0\\
 u_2&2&0&3
\end{array}
\]
Suppose that the one-project restriction in \Cref{def:plus} is omitted, equivalently removing \eqref{eq:oneproject} from the assignment formulation. No feasible outcome then satisfies the resulting property. For $W=\{b,c\}$, use $\beta=4$, $z_1=z_2=1$, and $y_a=x_{1a}=x_{2a}=2$, with all other variables zero. The total voter weight and normalized cost both equal two, and all remaining constraints hold. For $W=\{a,b\}$, use $\beta=3$, $z_2=y_c=x_{2c}=1$, with all other variables zero; exchange the voters and $b,c$ for $W=\{a,c\}$. A smaller outcome is contained in one of these pairs, and the displayed violation for the pair remains valid when selected projects are removed. These cases exhaust the feasible outcomes.

Yet $\{b,c\}$ satisfies exact FJR: a singleton can afford at most one project and obtain utility at most three, while the full group cannot obtain minimum utility above three from any feasible set. The infeasibility above is therefore caused by the fractional requirements, not by FJR itself. Condition~\eqref{eq:oneproject} removes the displayed assignments because the missing project would take the assigned voter strictly above the target.

The same instance also distinguishes our guarantee from exact additive FJR. RBG selects $a$ at level four, charging each voter $1/2$, and then no further project is affordable. Its feasible completions satisfy FJR1+ by \Cref{thm:compute}, but none satisfies exact FJR: at least one voter has utility two and can afford her private project of utility three.
\end{example}

\begin{example}[FJR1+ is a strict strengthening]
\label{ex:strictplus}
Let $n=4$, $k=3$, with unit costs and ballots
\[
 A_1=\{a,d_1\},\quad A_2=\{a,d_2\},\quad
 A_3=\{a,d_3\},\quad A_4=\{d_1,d_2,d_3\}.
\]
The outcome $W=\{d_1,d_2,d_3\}$ satisfies FJR, and hence FJR1. A group containing voter four cannot complain, since that voter already has utility three and a proportionally affordable set contains at most three projects. A group drawn from the first three voters can afford at most two projects. It cannot give utility two to all its members: giving utility two to two distinct members needs $a$ and their two different private projects. A singleton cannot afford even one project. Targets at most one are already met.

Nevertheless, at level two, set $y_a=1$ and $x_{1a}=x_{2a}=x_{3a}=1$, with $z_1=z_2=z_3=1/2$, and all other variables zero. These assignments satisfy \eqref{eq:assignments}, while their voter weight $3/2$ is larger than their normalized cost $4/3$. Hence $W$ violates FJR1+. The example illustrates why the verification theorem is for a stronger requirement.
\end{example}

\begin{example}[Priceable Hare-FJR+ need not be Droop-FJR+]
\label{ex:droop}
Consider the following unit-cost PB instance with approval utilities, also used by \citet{Teh2026}. Let $n=4$, $k=2$, and
\[
 A_1=\{y,a,b\},\quad A_2=\{x,a,b\},\quad
 A_3=\{c,a\},\quad A_4=\{b\}.
\]
Take $W=\{x,y\}$. This outcome is priceable with $B'=4$: voter one pays one for $y$, voter two pays one for $x$, and the other voters do not pay. The remaining budget available to each of $a,b,c$ is one.

We verify Hare-FJR+ directly. At integer level one, only voters three and four can participate, and their approved sets are disjoint. Consequently,
\[
 \sum_i z_i\le y_a+y_b+y_c\le\sum_d y_d
 <2\sum_d y_d
\]
whenever the project weight is positive. At every integer level $h\ge2$, each project is approved by at most three voters, so
\[
 h\sum_i z_i=\sum_{i,d}x_{id}\le3\sum_d y_d,
 \quad
 \sum_i z_i\le\tfrac32\sum_d y_d<2\sum_d y_d.
\]
Neither case permits a Hare violation, whose normalized price is two.

For Droop, at level two use $y_a=1$, $x_{1a}=x_{2a}=x_{3a}=1$, and $z_1=z_2=z_3=1/2$, setting the other variables to zero. Their total voter weight is $3/2>4/3$, the Droop price of the one weighted project. This is a Droop-FJR+ violation. Thus the simultaneous conclusion of \Cref{thm:droop} is strictly stronger than Hare-FJR+ together with priceability.
\end{example}

\end{document}